%% file: partial_covers.tex
\documentclass{article}
\usepackage{fullpage}
  \usepackage{verbatim}
  \usepackage{authblk}
  \usepackage{stfloats}
  \usepackage{graphicx}
  \usepackage{bold-extra}
  \usepackage{enumerate}
  \usepackage{amsfonts,amsmath,amsthm}
  \usepackage{tikz}
  \usetikzlibrary{shapes, calc}
  \usepackage{rotating}
  
  \usepackage[algo2e, noend, noline, boxed]{algorithm2e}
  \SetKwComment{tcm}{\{}{\}}
   \newenvironment{myfunction}[2][htbp]
  {%
    \setlength{\algomargin}{.2cm}
    \begin{center}
    \begin{minipage}{#2}
    \begin{function}[#1]
    \small
     \let\Par=\par
       \def\par{\endgraf\vspace{.1cm}}
           \SetKw{To}{to}%
       \SetKw{Downto}{downto}%
           \SetKw{Or}{or}%
       \SetKwFor{Algo}{Function}{}{}%
      \vspace{.15cm}%
   }
   {%
     \let\par=\Par\end{function}%
     \end{minipage}%
     \end{center}%
   }
  \newenvironment{myalgorithm}[2][htbp]
  {%
    \setlength{\algomargin}{.2cm}
    \begin{center}
    \begin{minipage}{#2}
    \begin{algorithm2e}[#1]
    \small
     \let\Par=\par
       \def\par{\endgraf\vspace{.1cm}}
           \SetKw{To}{to}%
       \SetKw{Downto}{downto}%
           \SetKw{Or}{or}%
       \SetKwFor{Algo}{Algorithm}{}{}%
      \vspace{.15cm}%
   }
   {%
     \let\par=\Par\end{algorithm2e}%
     \end{minipage}%
     \end{center}%
   }

  \theoremstyle{theorem}
   \newtheorem{theorem}{Theorem}
   \newtheorem{observation}{Observation}
   \newtheorem{lemma}{Lemma}
  
  \newtheorem{claim}[theorem]{Claim}
  \theoremstyle{definition}
  \newtheorem{example}{Example}

  \newcommand{\ChangeList}{\mathit{ChangeList}}
  \newcommand{\Covered}{\mathit{Covered}}
  \newcommand{\CST}{\mathit{CST}}
  \newcommand{\sCST}{\mathit{sCST}}
  \newcommand{\children}{\mathit{children}}
  \renewcommand{\next}{\mathit{next}}
  \newcommand{\tree}{\mathit{tree}}
  \newcommand{\extend}{\Delta}
  \newcommand{\firstocc}{\mbox{\textit{first}}}
  \newcommand{\lastocc}{\mbox{\textit{last}}}
  \newcommand{\Occ}{\mathit{Occ}}
  \renewcommand{\c}{\mathit{cv}}
  \newcommand{\Oh}{\mathcal{O}}
  
  \newcommand{\Lift}{\mathit{Lift}}
  \newcommand{\LocalCorrect}{\mathit{LocalCorrect}}
  \newcommand{\Dist}{\mathit{Dist}}
  \newcommand{\List}{\mathit{List}}
  \newcommand{\MultiInsert}{\mathit{MultiInsert}}
  \newcommand{\MultiPred}{\mathit{MultiPred}}
  \newcommand{\MultiSucc}{\mathit{MultiSucc}}
  \newcommand{\PP}{\mathcal{P}}
  \newcommand{\union}{\mathit{Union}}
  \newcommand{\find}{\mathit{Find}}  
  \newcommand{\EE}{\mathcal{E}}
  \newcommand{\PC}{{\textsc{PartialCovers}}}
  \newcommand{\APC}{{\textsc{AllPartialCovers}}}
  
  \newcommand{\Leaves}{\mathit{Leaves}}  

  \title{
    Fast Algorithm for Partial Covers in Words\footnote{
      A preliminary version of this work appeared in the Proceedings of
      the Twenty-Fourth Annual Symposium on Combinatorial Pattern Matching, pp. 177--188, 2013.
    }
  }
  
  \date{}
  \author[1]{Tomasz Kociumaka\thanks{
    Supported by Polish budget funds for science in 2013-2017 as a research project under the `Diamond Grant' program.
    }}
\author[1]{Jakub Radoszewski\thanks{The author receives financial support of Foundation for Polish Science.}}
\author[1,2]{Wojciech Rytter\thanks{Supported by grant no.\ N206 566740 of the National Science Centre.}}
\author[3]{Solon P. Pissis\thanks{Supported by the NSF--funded iPlant Collaborative (NSF grant \#DBI-0735191).}}
\author[1]{Tomasz Wale\'n\thanks{Supported by Iuventus Plus grant (IP2011 058671) of the Polish Ministry of Science and Higher Education.}}

\affil[1]{Faculty of Mathematics, Informatics and Mechanics, University of Warsaw, Poland, \texttt{[kociumaka, jrad, rytter, walen]@mimuw.edu.pl}}
\affil[2]{Faculty of Mathematics and Computer Science, Copernicus University, Toru\'n, Poland}
\affil[3]{Department of Informatics, King's College London, UK, \texttt{solon.pissis@h-its.org}}

\begin{document}
\maketitle
  \begin{abstract}
    A factor $u$ of a word $w$ is a \emph{cover} of $w$ if every position in $w$ lies within some occurrence of $u$ in $w$.
    A word $w$ covered by $u$ thus generalizes the
    idea of a \emph{repetition}, that is, a word composed of exact concatenations of $u$.
    In this article we introduce a new notion of
    $\alpha$-\emph{partial cover}, which can be viewed as a relaxed variant
    of cover, that is, a factor covering at least $\alpha$ positions in $w$.
    We develop a data structure of $\Oh(n)$ size (where $n=|w|$) that can be constructed in $\Oh(n\log n)$ time
    which we apply to compute all shortest $\alpha$-partial covers for a given $\alpha$.
    We also employ it for an $\Oh(n\log n)$-time algorithm computing a shortest $\alpha$-partial cover
    for each $\alpha=1,2,\ldots,n$.
  \end{abstract}

    \section{Introduction}
    The notion of periodicity in words and its many variants have been well-studied in numerous
    fields like combinatorics on words, pattern matching, data compression, automata theory,
    formal language theory, and molecular biology (see \cite{DBLP:journals/tcs/CrochemoreIR09}).
    However the classic notion of
    periodicity is too restrictive to provide a description of a word such as
    $\texttt{abaababaaba}$, which is covered by copies of $\texttt{aba}$, yet not exactly periodic.
    To fill this gap, the idea of \emph{quasiperiodicity} was introduced~\cite{Apo93}.
    In a periodic word, the occurrences of the period do not overlap. In contrast,
    the occurrences of a quasiperiod in a quasiperiodic word may overlap.
    Quasiperiodicity thus enables the
    detection of repetitive structures that would be ignored by the classic characterization
    of periods.

    The most well-known formalization of quasiperiodicity is the cover of word.
    A factor $u$ of a word $w$ is said to be
    a \emph{cover} of $w$ if $u\ne w$, and every position in $w$ lies within
    some occurrence of $u$ in $w$.
    Equivalently, we say that $u$ \emph{covers} $w$. Note that a cover of $w$ must also be
    a \emph{border} --- both prefix and suffix --- of $w$.
    Thus, in the above example, $\texttt{aba}$ is the shortest cover of $\texttt{abaababaaba}$.

    A linear-time algorithm for computing the shortest cover of a word was proposed
    by Apostolico et al.~\cite{DBLP:journals/ipl/ApostolicoFI91}, and a linear-time
    algorithm for computing all the covers of a word was proposed by
    Moore \& Smyth~\cite{DBLP:journals/ipl/MooreS94}. Breslauer~\cite{DBLP:journals/ipl/Breslauer92}
    gave an online linear-time algorithm computing the {\em minimal cover array} of a word --- a data 
    structure specifying the shortest cover of every prefix of the word.
    Li \& Smyth~\cite{DBLP:journals/algorithmica/LiS02} provided a linear-time algorithm
    for computing the {\em maximal cover array} of a word, and showed that, analogous to
    the border array~\cite{AlgorithmsOnStrings}, it actually determines the structure of {\em all} the covers of every prefix of the word.

    A known extension of the notion of cover is the notion of \emph{seed}.
    A seed is not necessarily aligned with the ends of the word being covered,
    but is allowed to overflow on either side.
    More formally, a word $u$ is a seed of $w$ if $u$ is a factor of $w$ and $w$ is a factor of some
    word $y$ covered by $u$.
    Seeds were first introduced by Iliopoulos, Moore, and Park~\cite{DBLP:journals/algorithmica/IliopoulosMP96}.
    A linear algorithm for computing the shortest seed of a word was given
    by Kociumaka et al.~\cite{DBLP:conf/soda/KociumakaKRRW12}.

    Still it remains unlikely that an arbitrary word, even over the binary alphabet,
    has a cover (or even a seed).
    For example, $\texttt{abaaababaabaaaababaa}$ is a word that not only has no cover, 
    but whose every prefix also has no cover.    
    In this article we provide a natural form of quasiperiodicity.
    We introduce the notion of \emph{partial covers}, that is, factors covering at least a given number of positions in $w$.
    Recently, Flouri et al.~\cite{Flouri2013102} suggested a related notion of \emph{enhanced covers}
    which are additionally required to be borders of the word.
    
    Partial covers can be viewed as a relaxed variant of covers
    alternative to approximate covers~\cite{SimParkKimLee}.
    The approximate covers require each position to lie within an
    approximate occurrence of the cover. This allows for small irregularities
    within each fragment of a word.
    On the other hand partial covers require exact occurrences but drop
    the condition that all positions need to be covered. This allows some
    fragments to be completely irregular as long as the total length of such
    fragments is small.
    The significant advantage of partial covers is that they enjoy more
    combinatorial properties, and consequently the algorithms solving the most
    natural problems are much more efficient than those concerning approximate
    covers, where the time complexity rarely drops below quadratic and
    some problems are even NP-hard. 
    
    Let $\Covered(u,w)$ denote the number of positions in $w$ covered by occurrences of the word $u$ in $w$; 
    we call this value the \emph{cover index} of $u$ within $w$.
    For example, $\Covered(\texttt{aba},\texttt{aababab})=5$. We primarily
    focus on the following two problems, but the tools we develop can be
    used to answer a number of questions concerning partial covers,
    some of which are discussed in the Conclusions.
    \vskip 0.3cm \noindent {\PC\ \bf problem}

    \noindent \hspace*{0.2cm}
    {\bf Input:} a word $w$ of length $n$ and a positive integer $\alpha\le n$.

    \noindent \hspace*{0.2cm}
    {\bf Output:} all shortest factors $u$ such that $\Covered(u,w)\ge \alpha$.
    \vskip 0.3cm \noindent
    Each factor given in the output is represented by the first and the last starting position
    of its occurrence in $w$.

    \begin{example}
      Let $w=\texttt{bcccacccaccaccb}$ and $\alpha=11$.
      Then the only shortest $\alpha$-partial covers are $\texttt{ccac}$ and $\texttt{cacc}$.
    \end{example}

    \noindent {\APC\ \bf problem}

    \noindent \hspace*{0.2cm}
    {\bf Input:} a word $w$ of length $n$.

    \noindent \hspace*{0.2cm}
    {\bf Output:} for all $\alpha=1,\ldots,n$, a shortest factor $u$
    such that $\Covered(u,w)\ge \alpha$.

    \vskip 0.3cm
    \noindent{\bf Our contribution.} The following summarizes our main result.
    \begin{theorem}\label{thm:main}
      The \PC\ and \APC\ problems can be solved
      in $\Oh(n\log n)$ time and $\Oh(n)$ space.
    \end{theorem}

    We extensively use suffix trees, for an exposition see \cite{AlgorithmsOnStrings,Jewels}.
    A suffix tree of a word is a compact trie of its suffixes, the nodes of the trie which become
    nodes of the suffix tree are called {\em explicit} nodes, while the other
    nodes are called {\em implicit}. 
    Each edge of the suffix tree can be viewed as an {\em upward}
    maximal path of implicit nodes starting with an explicit node.
    Moreover, each node belongs to a unique path of that kind.
    Then, each node of the trie can be represented in the suffix tree
    by the edge it belongs to and an index within the corresponding path.
    Each factor of the word corresponds to an explicit or implicit
    node of the suffix tree.
    A representation of this node is called the \emph{locus} of the factor.
    Our algorithm finds the loci of the shortest partial covers, it is then straightforward
    to locate an occurrence for each of them.

    \paragraph{\bf A Sketch of the Algorithm}
    The algorithm first augments the suffix tree of $w$, that is, a linear number of implicit
    extra nodes become explicit.
    Then, each node of the augmented tree is annotated with two integer values.
    They allow for determining the size of the covered area 
    for each implicit node by a simple formula, since limited to a single edge of the augmented suffix tree,
    these values form an arithmetic progression.
    This yields a solution to the \PC.
    For an efficient solution to the \APC\ problem, we additionally
    find the upper envelope of a number of line segments constructed from the arithmetic progressions.
    
    \paragraph{\bf Structure of the Paper}
    In Section~\ref{sec:CST} we formally introduce the augmented and annotated suffix tree
    that we call \emph{Cover Suffix Tree}.
    We show its basic properties and present its application for \PC\
    and \APC\ problems.
    Section~\ref{sec:CST_construction} is dedicated to the construction
    of the Cover Suffix Tree.
    Before that, Section~\ref{sec:fu} presents an auxiliary data structure being an extension
    of the classical Union/Find data structure; its implementation is given later, in Section~\ref{sec:details}.
    Additional applications of the Cover Suffix Tree are given in Sections~\ref{sec:by-products}
    and~\ref{sec:conclusions}.
    The former presents how the data structure can be used to compute all primitively rooted squares
    in a word and a linear-sized representation of all the seeds in a word.
    The latter contains a short discussion of variants of the \PC\ problem
    that can be solved in a similar way.

  \section{Augmented and Annotated Suffix Trees}\label{sec:CST}
  Let $w$ be a word of length $n$ over a totally ordered alphabet $\Sigma$.
  The suffix tree $T$ of $w$ can be constructed in $\Oh(n\log{|\Sigma|})$ time
  \cite{DBLP:conf/focs/Farach97,DBLP:journals/algorithmica/Ukkonen95}.
  For an explicit or implicit node $v$ of $T$, we denote by $\hat{v}$ the word obtained by spelling
  the characters on a path from the root to $v$.
  We also denote $|v|=|\hat{v}|$.
  As in most applications of the suffix tree,
  the leaves of $T$ play an auxiliary role and do not correspond to factors
  (actually they are suffixes of $w\#$, where $\# \notin \Sigma$).
  They are labeled with the starting positions of the suffixes of $w$.

  We introduce the \emph{Cover Suffix Tree} of $w$, denoted by $\CST(w)$,
  as an \emph{augmented} --- new nodes are added --- suffix tree in which the
  nodes are \emph{annotated} with information relevant to covers. $\CST(w)$ is similar to the data structure named \emph{Minimal Augmented Suffix Tree}
  (see \cite{DBLP:journals/algorithmica/ApostolicoP96,DBLP:conf/icalp/BrodalLOP02}).

  For a set $X$ of integers and $x\in X$, we define
  $$\next_X(x)=\min\{y\in X, y>x\},$$
  and we assume $\next_X(x)=\infty$ if $x=\max X$.
  By $\Occ(v,w)$ we denote the set of starting positions of occurrences of $\hat{v}$ in $w$.
  For any $i \in \Occ(v,w)$, we define:
  $$\delta(i,v) = \next_{\Occ(v,w)}(i)-i.$$
  Note that $\delta(i,v)=\infty$ if $i$ is the last occurrence of $\hat{v}$.
  Additionally, we define:
  $$\c(v)\,=\, \Covered(\hat{v},w),\ \ \Delta(v)\,=\,  \big|\left\{i\in
  \Occ(v,w)\, :\, \delta(i,v)\ge |v|\right\}\big|;$$ see, for example,
  Fig.~\ref{fig:c_Delta}.

    \begin{figure}[htb]
      \centering
      \input{fig1}
      \caption{\label{fig:c_Delta}
        Let $w=\texttt{bcccacccaccaccb}$ and let $v$ be the node corresponding to $\texttt{cacc}$.
        We have $\Occ(v,w)=\{4, 8, 11\}$, $\c(v) = 11$, $\Delta(v)=2$.
      }
    \end{figure}

  A word $u$ is called \emph{primitive} if $u=y^k$ for a word $y$ and an integer $k$ implies that $y=u$, and
  non-primitive otherwise.
  A square $u^2$ is called \emph{primitively rooted} if $u$ is primitive.

  \begin{observation}\label{obs:primitively_rooted_square}
    Let $v$ be a node in the suffix trie of $w$.
    Then $\hat{v}\hat{v}$ is a primitively rooted square in $w$
    if and only if there exists $i \in \Occ(v,w)$ such that $\delta(i,v)=|v|$.
  \end{observation}
  \begin{proof}
    Recall that, by the synchronization property of primitive words (see~\cite{AlgorithmsOnStrings}),
    $\hat{v}$ is primitive if and only if it occurs exactly twice in $\hat{v}\hat{v}$.

    $(\Rightarrow)$
    If $\hat{v}\hat{v}$ occurs in $w$ at position $i$ then $\delta(i,v)=|v|$.

    $(\Leftarrow)$ If $\delta(i,v)=|v|$ then obviously $\hat{v}\hat{v}$ occurs in $w$ at position $i$.
    Additionally, if $\hat{v}$ was not primitive then $\delta(i,v)<|v|$ would hold.
  \end{proof}

  In $\CST(w)$, we introduce additional explicit nodes called \emph{extra nodes}, 
  which correspond to halves of primitively rooted square factors of $w$.
  Moreover we annotate all explicit nodes (including extra nodes) with the values
  $\c,\Delta$; see, for example, Fig.~\ref{fig:CST_example}.
  The number of extra nodes is bounded by the number of distinct squares, which is linear \cite{DBLP:journals/jct/FraenkelS98},
  so $\CST(w)$ takes $\Oh(n)$ space.
 
  \begin{lemma}\label{lem:formula}
    Let $v_1,v_2,\ldots,v_k$ be the consecutive implicit nodes on the edge from an explicit node $v$ of $\CST(w)$
    to its explicit parent.
    Then for $1\le i \le k$ we have
    $$\c(v_i) = \c(v)-i\Delta(v),$$
    in particular $(\c(v_i))_{i=1}^k$ forms an arithmetic progression.
  \end{lemma}
  \begin{proof}
    Note that $\Occ(v_i,w) = \Occ(v,w)$, since otherwise $v_i$ would be an explicit node of $\CST(w)$.
    Also note that if any two occurrences of $\hat{v}$ in $w$ overlap, then the corresponding
    occurrences of $\hat{v_i}$ overlap.
    Otherwise, by Observation~\ref{obs:primitively_rooted_square},
    the path from $v$ to $v_i$ (excluding $v$) would contain an extra node.
    Hence, when we go up from $v$ (before reaching its parent) the size of the covered area
    decreases at each step by $\Delta(v)$.
  \end{proof}
   \begin{figure}[htb]
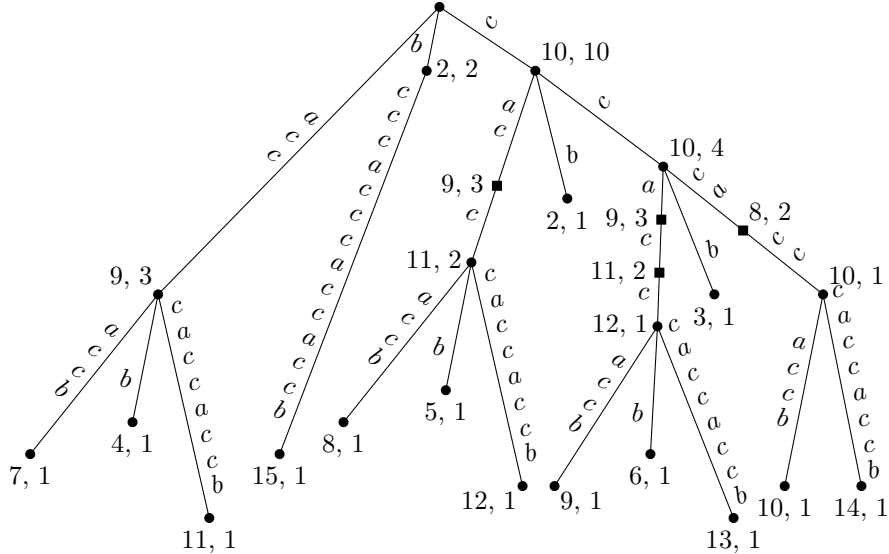

      \centering
      \include{fig2_suf_tree}
      \caption{
        $\CST(w)$ for $w=\texttt{bcccacccaccaccb}$.
        It contains four extra nodes that are denoted by squares in the figure.
        Each node is annotated with $\c(v), \Delta(v)$. Leaves are omitted for clarity.
      }\label{fig:CST_example}
    \end{figure}
  \begin{example}
    Consider the word $w$ from Fig.~\ref{fig:CST_example}.
    The word $\texttt{cccacc}$ corresponds to an explicit node of $\CST(w)$; we denote it by $v$.
    We have $\c(v)=10$ and $\Delta(v)=1$ since the two occurrences of the factor $\texttt{cccacc}$ 
    in $w$ overlap.
    The word $\texttt{cccac}$ corresponds to an implicit node $v'$ and $\c(v') = 10 - 1 = 9$.
    Now the word $\texttt{ccca}$ corresponds to an extra node $v''$ of $\CST(w)$.
    Its occurrences are adjacent in $w$ and $\c(v'')=8$, $\Delta(v'')=2$.
    The word $\texttt{ccc}$ corresponds to an implicit node $v'''$ and $\c(v''') = 8 - 2 = 6$.
  \end{example}

  \noindent
  As a consequence of Lemma~\ref{lem:formula} we obtain the following result.
  Recall that the locus of a factor $v$ of $w$, given by its start and end position
  in $w$, can be found in $\Oh(\log\log |v|)$ time \cite{DBLP:conf/cpm/KucherovNS12}.
  \begin{lemma}\label{lem:if_CST}
    Assume we are given $\CST(w)$.
    Then we can compute:
    \begin{enumerate}[(1)]
      \item\label{one} for any $\alpha$, the loci of the shortest $\alpha$-partial covers in linear time;
      \item\label{two} given the locus of a factor $u$ in the suffix tree
      $\CST(w)$, the cover index $\Covered(u,w)$ in $\Oh(1)$ time.
    \end{enumerate}
  \end{lemma}
  \begin{proof}
    Part (\ref{two}) is a direct consequence of Lemma~\ref{lem:formula}.
    As for part (\ref{one}), for each edge of $\CST(w)$, leading from $v$ to its parent $v'$,
    we need to find minimum $|v| \ge j > |v'|$ for which
    $\c(v)-\Delta(v) \cdot (|v|-j) \ge \alpha$.
    Such a linear inequality can be solved in constant time.
  \end{proof}
  Due to this fact the efficiency of the \PC\ problem
  relies on the complexity of $\CST(w)$ construction.
  In turn, the following lemma, also a consequence of Lemma~\ref{lem:formula},
  can be used to solve \APC\ problem
  provided that $\CST(w)$ is given.
  As a tool a solution to the geometric problem of upper envelope
  \cite{DBLP:journals/ipl/Hershberger89} is applied.
  \begin{lemma}\label{lem:all_if_CST}
    Assume we are given $\CST(w)$.
    Then we can compute the locus of a shortest $\alpha$-partial cover
    for each $\alpha=1,2,\ldots,n$ in $\Oh(n\log n)$ time and $\Oh(n)$ space.
  \end{lemma}
  \begin{proof}
    Consider an edge of $\CST(w)$ from $v$ to its parent $v'$ containing $k$ implicit nodes.
    For each such edge, we form a line segment on the plane connecting points
    $(|v|,\c(v))$ and $(|v|-k,\c(v)-k\cdot\Delta(v))$
    (if there are no implicit nodes on the edge, the line segment is a single point).
    Denote all such line segments obtained from $\CST(w)$ as $s_1,\ldots,s_m$, we have $m=\Oh(n)$.
    We consider the upper envelope $\EE$ of the set of these segments.
    Formally, if each $s_i$ connecting points $(x_i,y_i)$ and $(x'_i,y'_i)$, $x_i \le x'_i$,
    is interpreted as a linear function on a domain $[x_i,x'_i]$,
    $\EE$ is defined as a function $\EE : [1,n] \rightarrow [1,n]$ such that:
    $$\EE(x) = \max\{s_i(x)\,:\, i\in \{1,\ldots,m\}, x \in [x_i,x'_i]\}.$$
    Here we are actually interested in an \emph{integer envelope} $\EE'$, that is, $\EE$
    limited to integer arguments, see Fig.~\ref{fig:envelope}.
    By Lemma~\ref{lem:formula}, for any $j \in \{1,\ldots,n\}$,
    $\EE'(j)$ equals the maximum of $\Covered(u,w)$ over all factors $u$ of $w$
    such that $|u|=j$.
    A piecewise linear representation of $\EE$ can be computed in $\Oh(m\log{m})$ time and $\Oh(m)$ space
    \cite{DBLP:journals/ipl/Hershberger89}, therefore the function $\EE'$ for all its arguments
    can be computed in the same time complexity.

    \begin{figure}[htb]
      \centering
      \input{fig5_envelope}
      \caption{\label{fig:envelope}
        Line segments constructed as in Lemma~\ref{lem:all_if_CST} for the $\CST(w)$ from Fig.~\ref{fig:CST_example}.
        The marked points joined with a dashed polyline show the values of the integer upper envelope function $\EE'$.
        We infer from the graph that the lengths of the shortest $\alpha$-partial covers of $w$
        are as follows:
        1 for $\alpha \le 10$, 4 for $\alpha=11$, 5 for $\alpha=12$, and $\alpha$ for $\alpha \ge 13$.
      }
    \end{figure}
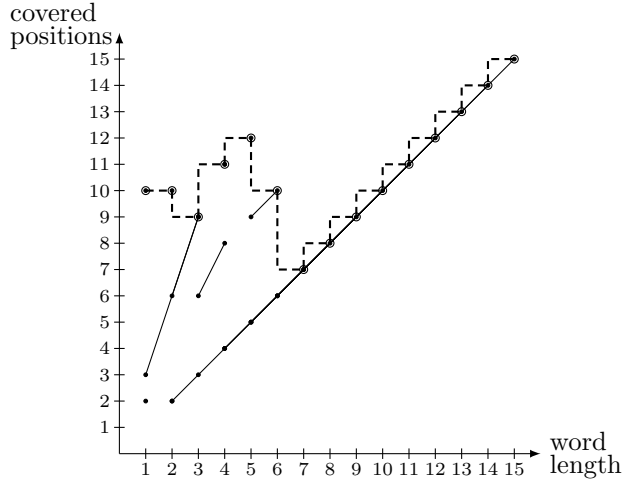

    Let us introduce a prefix maxima sequence for $\EE'$: $\mu_i = \max\{\EE'(j)\,:\,j \in \{1,\ldots,i\}\}$,
    with $\mu_0=0$.
    Note that $\mu_i$ is non-decreasing.
    If $\mu_i>\mu_{i-1}$ then the shortest $\alpha$-partial cover for all $\alpha \in (\mu_{i-1},\mu_i]$
    has length $i$.
    An example of such a partial cover can be recovered
    if we explicitly store the initial line segments used in the pieces of the representation of $\EE$.
    Thus the solution of the \APC\ problem can be obtained
    from the sequence $\mu_i$ in $\Oh(m)=\Oh(n)$ time.
  \end{proof}
  In the following two sections we provide an $\Oh(n\log n)$ time construction
  of $\CST(w)$.
  Together with Lemmas~\ref{lem:if_CST} and~\ref{lem:all_if_CST}, it
  yields Theorem~\ref{thm:main}.

  \section{Extension of Disjoint-Set Data Structure}\label{sec:fu}
  In this section we extend the classic disjoint-set data structure to compute
  the \emph{change lists} of the sets being merged, as defined below.
 \begin{figure}[ht]
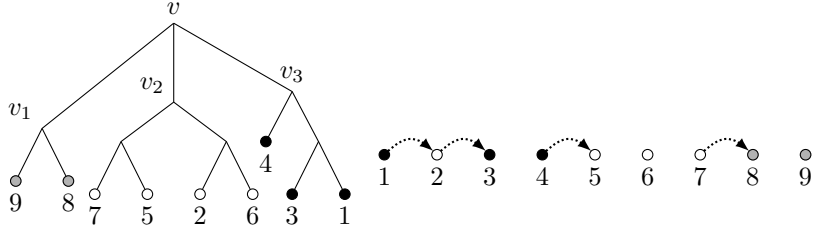

      \centering
      \include{fig4_merge_list}
      \caption{
        Let $\PP$ be the partition of $\{1,\ldots,9\}$ whose classes
        consist of leaves in the subtrees rooted at children of $v$,
        $\PP =\{\{1,3,4\},\{2,5,6,7\},\{8,9\}\}$,
        and let $\PP'=\{\{1,\ldots,9\}\}$.
        Then $\ChangeList(\PP,\PP') = \{(1,2),\,(2,3),\,(4,5),\,(7,8)\}$
        (depicted by dotted arrows).
      }\label{fig:ChangeList}
    \end{figure}
  First, let us extend the $\next$ notation.
  For a partition $\PP   = \{P_1,\ldots,P_k\}$ of $U=\{1,\ldots,n\}$, we
  define
  $$\next_{\PP}(x)=\next_{P_i}(x)\text{ where }x\in P_i.$$
  Now for two partitions $\PP  , \PP  '$ let us define the \emph{change list} (see also Fig.~\ref{fig:ChangeList}) by
  $$\ChangeList(\PP  ,\PP  ')= \{(x, \next_{\PP'}(x)): \next_{\PP}(x)\ne \next_{\PP'}(x)\}.$$

 \newcommand{\id}{id}
 We say that $(\PP,\id)$ is a partition of $U$ \emph{labeled} by $L$ if $\PP$ is a partition of $U$ and 
 $\id : \PP \to L$ is a one-to-one (injective) mapping.
 A label $\ell\in L$ is called \emph{active} if $\id(P)=\ell$ for some $P\in
 \PP$ and \emph{free} otherwise.
 \begin{lemma}\label{lem:main_data_structure}
    Let $n\le k$ be positive integers such that $k$ is of magnitude $\Theta(n)$.
    There exists a data structure of size $\Oh(n)$, which maintains a partition $(\PP,\id)$  of $\{1,\ldots,n\}$
    labeled by $L=\{1,\ldots,k\}$ and supports the following operations:
    \begin{itemize}
      \item $\find(x)$ for $x\in\{1,\ldots,n\}$ gives the label of $P\in \PP$ containing~$x$.
      \item $\union(I,\ell)$ for a set $I$ of active labels and a free label
      $\ell$ replaces all $P\in \PP$ with labels in $I$ by their set-theoretic union with the label $\ell$.
      The change list of the corresponding modification of $\PP$ is returned.
    \end{itemize}
    Initially $\PP$ is a partition into singletons with $\id(\{x\})=x$.
    Any valid sequence of $\union$ operations is performed in $\Oh(n\log n)$ time.
    A single $\find$ operation takes $\Oh(1)$ time.
  \end{lemma}
 
  Note that these are actually standard disjoint-set data structure operations except for the fact that
  we require $\union$ to return the change list.
  The technical proof of Lemma~\ref{lem:main_data_structure} is postponed until Section~\ref{sec:details}.

  \section{$\Oh(n \log n)$-time Construction of $\CST(w)$}\label{sec:CST_construction}
  The suffix tree of $w$ augmented with extra nodes is called the
  \emph{skeleton} of $\CST(w)$, which we denote by $\sCST(w)$.
  It could be constructed using the fact that all square factors of a word can be
  computed in linear time~\cite{DBLP:journals/jcss/GusfieldS04,Crochemore2013,DBLP:conf/spire/CrochemoreIKRRW10}.
  However, we do not need such a complicated machinery here.
  We will compute $\sCST(w)$ on the fly, simultaneously annotating the nodes with $\c$, $\Delta$.

  We introduce auxiliary notions related to covered area of nodes:
  $$\c_h(v)=\sum_{\substack{i\in \Occ(v,w)\\\delta(i,v)<h}} \delta(i,v),\quad\extend_h(v)=|\{i\in \Occ(v,w)\ :\, \delta(i,v) \ge h \}|.$$

  \begin{observation}
    $\c(v)\,=\, \c_{|v|}(v)+\extend_{|v|}(v)\cdot |v|,\,
    \Delta(v)=\extend_{|v|}(v).$
  \end{observation}

  \noindent
  In the course of the algorithm some nodes will have their values $c,\Delta$ already computed;
  we call them \emph{processed nodes}.
  Whenever $v$ will be processed, so will its descendants.
    
  The algorithm processes inner nodes $v$ of $\sCST(w)$ in the order of non-increasing height $h=|v|$.
  The height is not defined for leaves, so we start with $h=n+1$.
  Extra nodes are created on the fly using Observation~\ref{obs:primitively_rooted_square}
  (this takes place in the auxiliary $\Lift$ routine).

  We maintain the partition $\PP$ of $\{1,\ldots, n\}$ given by sets of leaves of subtrees
  rooted at \emph{peak nodes}.
  Initially the peak nodes are the leaves of $\sCST(w)$.
  Each time we process $v$ all its children are peak nodes.
  Consequently, after processing $v$ they are no longer peak nodes and $v$ becomes a new peak node.
  The sets in the partition are labeled with identifiers of the
  corresponding peak nodes.
  Recall that leaves are labeled with the starting positions of the corresponding suffixes. 
  We allow any labeling of the remaining nodes as long as each node of
  $\sCST(w)$ has a distinct label of magnitude $\Oh(n)$.
  For this set of labels we store the data structure of Lemma~\ref{lem:main_data_structure}
  to compute the change list of the changing partition.

    \begin{figure}[htpb]
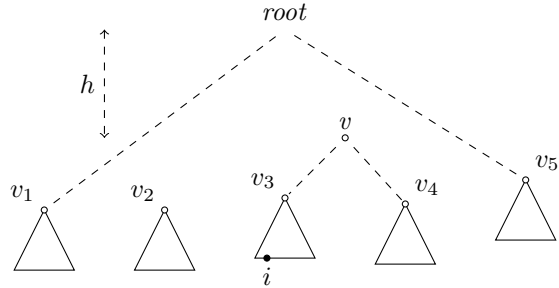

      \centering
      \include{drzewko}
      \caption{
        One stage of the algorithm, where the peak nodes are $v_1,\ldots,v_5$ while the
        currently processed node is $v$.
        If $i \in \List[d]$ and $v_3=\find(i)$, then $d=\delta(i,v_3)=Dist[i]$.
        The current partition is $\PP=\{\Leaves(v_1),\,
        \Leaves(v_2),$ $\, \Leaves(v_3),\, \Leaves(v_4), \, \Leaves(v_5)\}$.
        After $v$ is processed, the partition changes to
        $\PP=\{\Leaves(v_1),\, \Leaves(v_2),\, \Leaves(v),\, \Leaves(v_5)\}$.
        The $\union$ operation merges $\Leaves(v_4), \Leaves(v_3)$
        and returns the corresponding change list.
      }\label{fig:tree}
    \end{figure}

  \noindent
  We maintain the following technical invariant (see Fig.~\ref{fig:tree}).

  \medskip
  \noindent
	{\bf Invariant$(h)$:}
  \begin{enumerate}[\bf (A)]
    \item
    For each peak node $z$ we store:
    $$\c'[z]=\c_h(z),\,\Delta'[z]=\extend_h(z).$$
    \item
    For each $i \in \{1,\ldots,n\}$ we store $\Dist[i]=\delta(i,\,\find(i))$.
    \item
    For each $d<h$ we store $\List[d]\,=\, \{i\,:\, \Dist[i]=d\}$.
  \end{enumerate}

  \noindent
  We use two auxiliary routines.
  The $\Lift$ operation updates $\c'$ and $\Delta'$ values when $h$ decrements.
  It also creates all extra nodes of depth $h$.
  The $\LocalCorrect$ operation is used for updating $\c'$ and $\Delta'$ values for children
  of the node $v$.
  The $\Dist$ and $\List$ arrays are stored to enable efficient implementation
  of these two routines.

   \newcommand{\algname}{\textsc{ComputeCST}}
  \begin{myalgorithm}[H]{11.8 cm}
    \Algo{\algname(w)}
    {
      $T :=$\,suffix tree of $w$\;
      $\PP$ := partition of $\{1,\ldots,n\}$ into singletons\;
      \lForEach{$v:$ a leaf of $T$}{$\c'[v]:=0$, $\Delta'[v]:=1$}\;
      \For{$h:=n+1$ \KwSty{downto} $0$}{
        $\Lift(h)$\;
        \{Now part (A) of Invariant($h$) is satisfied\}\\
        \ForEach{$v:$ an inner node of $T$, $|v|=h$}{
          $\c'[v]:=\, \sum_{u\in \children(v)}\, \c'[u]$\;
          $\Delta'[v]:=\, \sum_{u\in \children(v)}\Delta'[u]$\;
          $\ChangeList(v)$ := $\union(\children(v),v)$\\
            \lForEach{$(p,q) \in \ChangeList(v)$}{$\LocalCorrect(p,q,v)$}\;
            $\c[v]\,:=\, \c'[v]+\Delta'[v]\cdot |v|$; $\Delta[v]:=\Delta'[v]$\;
        }
      }
      \Return{$T$ together with values of  $\c, \Delta$\;}
    }
  \end{myalgorithm}

    \paragraph{\bf Description of the $\Lift(h)$ Operation}
    The procedure $\Lift$ plays an important preparatory role in processing the current node.
    According to part (A) of our invariant, for all peak nodes $z$ we know the values:
    $\c'[z]=\c_{h+1}(z),\,\Delta'[z]=\extend_{h+1}(z).$
    Now we have to change $h+1$ to $h$ and guarantee validity of the invariant:
    $\c'[z]=\c_{h}(z),\,\Delta'[z]=\extend_{h}(z).$
    This is exactly how the following operation updates $\c'$ and $\Delta'$.

    It also creates all extra nodes of depth $h$ that were not explicit nodes of
    the suffix tree.
    By Observation~\ref{obs:primitively_rooted_square}, if $i\in \List[h]$
    then at position $i$ in $w$ there is an occurrence of a primitively rooted square
    of half length $h$.
    Consequently, an extra node corresponding to this occurrence is created
    in the $\Lift$ operation.

    \begin{myfunction}[H]{9 cm}
      \Algo{$\Lift(h)$}{
        \ForEach{$i$ \KwSty{in} $\List[h]$}{
          $v:=\find(i)$\;
          $\Delta'[v]:=\Delta'[v]+1$; $\c'[v]:=\c'[v]-h$\;
          \If{$|\mathit{parent}(v)|<h$}{
            Create a node of depth $h$ on the edge from $\mathit{parent}(v)$ to $v$
          }
        }
      }
    \end{myfunction}

    \paragraph{\bf Description of the $\LocalCorrect(p,q,v)$ Operation}
    Here we assume that $\hat{v}$ occurs at positions $p<q$ and that these are consecutive occurrences.
    Moreover, we assume that these occurrences are followed by distinct characters, i.e. $(p,q)\in \ChangeList(v)$. 
   	The $\LocalCorrect$ procedure updates $\Dist[p]$ to make part (B) of the invariant hold for $p$ again.
   	The data structure $\List$ is updated accordingly so that (C) remains satisfied.
    \begin{myfunction}[H]{9.5 cm}
      \Algo{$\LocalCorrect(p,q,v)$}{
        $d:=q-p$;\ $d':=\Dist[p]$\;
        \lIf{$d'<|v|$}{
          $\c'[v]:=\c'[v]-d'$
          } \lElse{
          $\Delta'[v]:=\Delta'[v]-1$}\;
        \lIf{$d<|v|$}{
          $\c'[v]:=\c'[v]+d$ }
        \lElse{
          $\Delta'[v]:=\Delta'[v]+1$}\;
        $\Dist[p]:=d$\;
         $\mathit{remove}(i, \List[d'])$;\ $\mathit{insert}(i, \List[d])$\;
      }
    \end{myfunction}

    \paragraph{\bf Complexity of the Algorithm}
    In the course of the algorithm we compute $\ChangeList(v)$ for each $v\in T$.
    Due to Lemma~\ref{lem:main_data_structure} we have:
    $$\sum_{v \in T}\; |\ChangeList(v)|\; =\; \Oh(n\log n).$$
    Consequently we perform $\Oh(n\log n)$ operations $\LocalCorrect$.
    In each of them at most one element is added to a list $\List[d]$ for some $d$.
    Hence the total number of insertions to these lists is also $\Oh(n\log n)$.

    The cost of each operation $\Lift$ is proportional to the total size of
    the list $\List[h]$ processed in this operation.
    For each $h$ the list $\List[h]$ is processed once and
    the total number of insertions into lists is $\Oh(n\log n)$,
    therefore the total cost of all operations $\Lift$ is also $\Oh(n\log n)$.
    This proves the following fact which, together with Lemmas~\ref{lem:if_CST} and~\ref{lem:all_if_CST},
    implies our main result (Theorem~\ref{thm:main}).

    \begin{lemma}
      Algorithm \algname~constructs $\CST(w)$ in $\Oh(n\log n)$ time and $\Oh(n)$ space, where $n=|w|$.
    \end{lemma}

  \section{Implementation Details}\label{sec:details}
  In this section we give a proof of Lemma~\ref{lem:main_data_structure}.  
  We use an approach similar to Brodal and Pedersen \cite{DBLP:conf/cpm/BrodalP00}
  (who use the results of \cite{DBLP:journals/jacm/BrownT79})
  originally devised for computation of maximal quasiperiodicities.

  Theorem 3 of \cite{DBLP:conf/cpm/BrodalP00} states that
  a subset $X$ of a linearly ordered universe can be stored in a height-balanced
  tree of linear size supporting the following operations:
  \begin{description}
    \item[\quad $X.\MultiInsert(Y)$:] insert all elements of $Y$ to $X$,
    \item[\quad $X.\MultiPred(Y)$:] return all $(y,x)$ for $y\in Y$ and $x=\max\{z\in X, z<y\}$,
    \item[\quad $X.\MultiSucc(Y)$:] return all $(y,x)$ for $y\in Y$ and $x=\min\{z\in X, z>y\}$,
  \end{description}
  in $O\left(|Y|\max\left(1, \log\frac{|X|}{|Y|}\right)\right)$ time.

  \smallskip

  In the data structure we store each $P\in \PP$ as a height-balanced tree.
  Additionally, we store several auxiliary arrays, whose semantics follows.
  For each $x\in \{1,\ldots,n\}$ we maintain a value $\next[x]=\next_\PP(x)$
  and a pointer $\tree[x]$ to the tree representing $P$ such that $x\in P$.
  For each $P\in \PP$ (technically for each tree representing $P\in \PP$) we
  store $\id[P]$ and for each $\ell\in L$ we store $\id^{-1}[\ell]$, a pointer to the
  corresponding tree (null for free labels).

  Answering $\find$ is trivial as it suffices to follow the $\tree$
  pointer and return the $\id$ value.
  The $\union$ operation is performed according
  to the pseudocode given below (for brevity we write $P_{i}$ instead of
  $\id^{-1}[i]$).

  \begin{claim}
    The $\union$ operation correctly computes the change list and updates the
    data structure.
  \end{claim}
  \begin{proof}
    In the $\union$ operation for sets $P_i$, $i \in I$, we find the largest set $P_{i_0}$
    and $\MultiInsert$ all the elements of the remaining sets to $P_{i_0}$.
    If $(a,b)$ is in the change list, then $a$ and $b$ come from
    different sets $P_i$, in particular at least one of them does not come from
    $P_{i_0}$. Depending on which one it is, the pair $(a,b)$ is found by
    $\MultiPred$ or $\MultiSucc$ operation.
    While computing $C$, the table $\next$ is not updated yet
    (i.e. corresponds to the state before $\union$ operation) while $S$ is already
    updated. Consequently the pairs inserted to $C$ indeed belong to the change
    list. 
    Once $C$ is proved to be the change list, it is clear that $\next$ is
    updated correctly.  For the other components of the data structure,
    correctness of updates is evident.
    
  \end{proof}

  \begin{myfunction}[H]{10.8 cm}
  \Algo{$\union(I,\ell)$}{
      $i_0 := \textrm{argmax}\{|P_i|: i\in I\}$\;
      $S := P_{i_0}$\;
     \ForEach{$i\in I\setminus\{i_0\}$}{
      \lForEach{$x\in P_i$}{$\tree[x] := S$}\;
      $S.\MultiInsert(P_i)$\;
     }
     $C := \emptyset$\;
     \ForEach{$i\in I\setminus\{i_0\}$}{\ForEach{$(b,a)\in S.\MultiPred(P_i)$}{
      \lIf{$\next[a]\ne b$}{$C := C \cup \{(a,b)\}$}\;
     }
     \ForEach{$(a,b)\in S.\MultiSucc(P_i)$}{
      \lIf{$\next[a]\ne b$}{$C := C \cup \{(a,b)\}$}\;
      }
      $\id^{-1}[i] := \textrm{null}$\;
     }
     $\id^{-1}[i_0] := \textrm{null}$\;
     $\id[S]:=\ell$; $\id^{-1}[\ell]:= S$\;
     \lForEach{$(x,y)\in C$}{$\next[x] := y$}
     \Return{$C$}\;
  }
  \end{myfunction}
  
  \begin{claim}
    Any sequence of $\union$ operations takes $\Oh(n\log n)$ time in total.
  \end{claim}
  \begin{proof}
    Let us introduce a potential function $\Phi(\PP) = \sum_{P\in \PP} |P|\log |P|$.
    We shall prove that the running time of a single $\union$ operation
    is proportional to the increase in potential.
    Clearly 
    $$ 0 \le \Phi(\PP) = \sum_{P\in \PP} |P|\log |P| \le \sum_{P\in \PP} |P|\log n
    = n\log n,$$
    so this suffices to obtain the desired $\Oh(n\log n)$ bound.
    
    Let us consider a $\union$ operation that merges partition classes of sizes
    $p_1\ge p_2 \ge \ldots \ge p_k$ to a single class of size $p = \sum_{i=1}^k
    p_i$.
    The most time-consuming steps of the algorithm are the operations on
    height-balanced trees, which, for single $i$, run in $O\left(\max\left(p_i,
    p_i\log\tfrac{p}{p_i}\right)\right)$ time. 
    These operations are not performed
    for the largest set and for the remaining ones we have $p_i < \frac{1}{2} p$
    (i.e. $\log{\frac{p}{p_i}}\ge 1$). This lets us bound the time complexity
    of the $\union$ operation as follows:
    $$\sum_{i=2}^k
    \max\left(p_i, p_i\log\tfrac{p}{p_i}\right) = \sum_{i=2}^k
    p_i\log\tfrac{p}{p_i} \le \sum_{i=1}^k p_i\log\tfrac{p}{p_i} = 
    \sum_{i=1}^k p_i(\log p - \log p_i)  = p\log p - \sum_{i=1}^k p_i\log p_i,$$
    which is equal to the increase in potential.
  \end{proof}

  \section{By-Products of Cover Suffix Tree}\label{sec:by-products}
  In this section we present two additional applications of the Cover Suffix Tree.
  We show that, given $\CST(w)$ (or $\CST$ of a word that can be obtained from $w$
  in a simple manner), one can compute in linear time all distinct primitively rooted squares
  in $w$ and a linear representation of all the seeds of $w$, in particular, the shortest
  seeds of $w$.
  This shows that constructing this data structure is \emph{at least as hard}
  as computing all primitively rooted squares and seeds. While there are linear-time algorithms
  for these problems \cite{DBLP:journals/jcss/GusfieldS04,DBLP:conf/focs/KolpakovK99,Crochemore2013} and \cite{DBLP:conf/soda/KociumakaKRRW12}, they
  are all complex and rely on the combinatorial properties specific to the repetitive structures they seek for.

  \begin{theorem}
    Assume that the Cover Suffix Tree of a word of length $n$ can be computed in $T(n)$ time.
    Then all distinct primitively rooted squares in a word $w$ of length $n$ can be computed in
    $T(2n)$ time.
  \end{theorem}
  \begin{proof}
    Let $\mathtt{0} \notin \Sigma$ be a special symbol.
    Let $\varphi:\Sigma^*\rightarrow(\Sigma\cup\{\mathtt{0}\})^*$ be a morphism such that
    $\varphi(c)=\mathtt{0}c$ for any $c \in \Sigma$.
    We consider the word $w'=\varphi(w)\mathtt{0}$, that is, the word $w$ with $\mathtt{0}$-characters
    inserted at all its inter-positions, e.g.\ if $w=\mathtt{aabab}$ then $w'=\mathtt{0a0a0b0a0b0}$.
    
    Let us consider the set of explicit non-branching nodes of $\CST(w')$
    and select among them the nodes corresponding to even-length factors of $w'$
    starting with the symbol $\mathtt{0}$.
    It suffices to note that there is a one-to-one correspondence between these
    nodes and the halves of primitively rooted squares in $w$.
  \end{proof}

  \begin{figure}[htpb]
  \begin{center}
  \input{fig6_seed}
  \caption{Seed of string \texttt{aaabaabaabaaabaaba}.}
  \label{fig:seed}
  \end{center}
  \end{figure}

  Recall that a word $u$ is a seed of $w$ if $u$ is a factor of $w$ and $w$ is a factor of some
  word $y$ covered by $u$, see Fig.~\ref{fig:seed}.
  The following lemma states that the set of all seeds of $w$ has a representation of $\Oh(n)$ size, where $n=|w|$.
  This representation enables, e.g., simple computation of all shortest seeds of the word.
  By a range on a edge of a suffix tree we mean a number of consecutive nodes on this edge
  (obviously at most one of these nodes is explicit).
  Let $w^R$ denote the reverse of the word $w$.
  
  \begin{lemma}[\cite{DBLP:journals/algorithmica/IliopoulosMP96,DBLP:conf/soda/KociumakaKRRW12}]\label{lem:seeds_representation}
    The set of all seeds of $w$ can be split into two disjoint classes.
    The seeds from one class form a single (possibly empty) range on each edge of the suffix tree of $w$,
    while the seeds from the other class form a range on each edge of the suffix tree of $w^R$.
  \end{lemma}


  We will show that given $\CST(w)$ and $\CST(w^R)$ we can compute the representation
  of all seeds from Lemma~\ref{lem:seeds_representation} in $\Oh(n)$ time.
  Let us recall auxiliary notions of quasiseed and quasigap, see \cite{DBLP:conf/soda/KociumakaKRRW12}.

  By $\firstocc(u)$ and $\lastocc(u)$ let us denote $\min \Occ(u)$ and $\max \Occ(u)$, respectively.
  We say that $u$ is a \emph{complete cover} in $w$ if $u$ is a cover of
  the word $w[\firstocc(u),\lastocc(u)+|u|-1]$.
  The word $u$ is called a \textit{quasiseed} of $w$ if $u$ is a complete cover in $w$,
  $\firstocc(u) < |u|$ and $n+1 - \lastocc(u) < 2|u|$.
  Alternatively, $w$ can be decomposed into $w=xyz$, where $|x|,|z|<|u|$ and $u$ is a cover of $y$.

  All quasiseeds of $w$ lying on the same edge of the suffix tree with lower explicit endpoint $v$
  form a range with the lower explicit end of the range located at $v$.
  The length of the upper end of the range is denoted as $\textsf{quasigap}(v)$.
  If the range is empty, we set $\textsf{quasigap}(v)=\infty$.
  Thus a representation of all quasiseeds of a given word can be provided using only 
  the quasigaps of explicit nodes in the suffix tree.
  It is known that computation of quasiseeds is the hardest part of an algorithm computing seeds:

  \begin{lemma}[\cite{DBLP:journals/algorithmica/IliopoulosMP96,DBLP:conf/soda/KociumakaKRRW12}]\label{lem:quasigaps}
    Assume quasigaps of all explicit nodes of suffix trees of $w$ and $w^R$ are known.
    Then a representation of all seeds of $w$ from Lemma~\ref{lem:seeds_representation}
    can be found in $\Oh(n)$ time.
  \end{lemma}

  It turns out that the auxiliary data in $\CST(w)$ and $\CST(w^R)$ enable constant-time computation
  of quasigaps of explicit nodes.
  By Lemma~\ref{lem:quasigaps} this yields an $\Oh(n)$ time algorithm for computing a representation
  of all the seeds of $w$.
  This is stated formally in the following theorem.

  \begin{theorem}\label{thm:seeds}
    Assume that the Cover Suffix Tree of a word of length $n$ can be computed in $T(n)$ time.
    Given a word $w$ of length $n$, one can compute a representation of all seeds of $w$
    from Lemma~\ref{lem:seeds_representation} in $T(n)$ time.
    In particular, all the shortest seeds of $w$ can be computed within the same time complexity.
  \end{theorem}
  \begin{proof}
    We show how to compute quasigaps for all explicit nodes of $\CST(w)$.
    The computation for $\CST(w^R)$ is symmetric.
    Note that $\CST(w)$ may contain more explicit nodes that the suffix tree of the word.
    In this case, the results from any maximal sequence of edges connected by non-branching explicit nodes
    in $\CST(w)$ need to be merged into a single range on the corresponding edge of the suffix tree.

    By the definition of $\c(v)$, an explicit node $v$ of $\CST(w)$ is a complete cover in $w$
    if the following condition holds:
    $$\c(v) = \lastocc(v) - \firstocc(v) + |v|.$$
    Thus for checking whether an explicit node $v$ of $\CST(w)$ is a quasiseed of $w$
    it suffices to check whether this condition and the following equalities hold:
    $$\firstocc(v) < |v|,\quad n+1 - \lastocc(v) < 2|v|.$$
    If $v$ is not a quasiseed of $w$, we have $\textsf{quasigap}(v)=\infty$, otherwise
    we can assume that $\textsf{quasigap}(v) \le |v|$.

    \begin{example}
      Consider the word $w$ from Fig.~\ref{fig:CST_example}, $n=15$.
      The word $\texttt{cacc}$ corresponds to an explicit node of $\CST(w)$; we denote it by $v$.
      We have $\c(v)=11$, $\firstocc(v)=4$, $\lastocc(v)=11$, and $\lastocc(v) - \firstocc(v) + |v|=11$.
      Therefore $\texttt{cacc}$ is a quasiseed of $w$, see also Fig.~\ref{fig:c_Delta}.
    \end{example}

    By Lemma~\ref{lem:formula}, the condition for any node on the edge ending at $v$ to be
    a complete cover in $w$ is very simple:
    $$\Delta(v) = 1.$$
    Assume this condition is satisfied and consider any implicit node $v'$ on this edge.
    Then $v'$ is a quasiseed if both inequalities:
    $$\firstocc(v) < k\quad\mbox{and}\quad n+1 - \lastocc(v) < 2k$$
    are satisfied.
    Thus in this case
    $$\textsf{quasigap}(v)=\max(\firstocc(v)+1,\,\lceil (n-\lastocc(v)+2)/2 \rceil,\,|\mathit{parent}(v)|+1).$$
  
    \begin{example}
      Consider the word $w$ from Fig.~\ref{fig:CST_example}.
      The word $\texttt{cccacc}$ corresponds to an explicit node of $\CST(w)$; we denote it by $v$.
      We have $\c(v)=10$, $\firstocc(v)=2$, $\lastocc(v)=6$, and $\lastocc(v) - \firstocc(v) + |v|=10$.
      Therefore $\texttt{cccacc}$ is a quasiseed of $w$. 
      Since $\Delta(v)=1$, $\textsf{quasigap}(v)$ could be smaller than 6.
      However, $\lceil (n-\lastocc(v)+2)/2 \rceil=6$ and the above formula yields $\textsf{quasigap}(v)=6$.
    \end{example}

    \noindent
    This concludes a complete set of rules for computing $\textsf{quasigap}(v)$ for explicit nodes of $\CST(w)$.
  \end{proof}

  \section{Conclusions}\label{sec:conclusions}
  We have presented an algorithm which constructs a data structure, called the
  {\em Cover Suffix Tree}, in $\Oh(n\log n)$ time and $\Oh(n)$ space.
  The Cover Suffix Tree has been developed in order to solve the
  \PC\ and \APC\ problem in $\Oh(n)$ and $\Oh(n\log n)$ time,
  respectively, but it also gives a well-structured description of the cover indices of all factors.
  Consequently, various questions related to partial covers can be answered efficiently.
  For example, with the Cover Suffix Tree one can solve in linear time a
  problem inverse to \PC:
  find a factor of length between $l$ and $r$ that maximizes the number of positions covered.
  Also a similar problem to \APC\ problem,
  to compute for all lengths $l=1,\ldots,n$ the maximum number of positions covered by a factor of length $l$,
  can be solved in $\Oh(n\log n)$ time.
  This solution was actually given implicitly in the proof of Lemma~\ref{lem:all_if_CST}.

  An interesting open problem is to reduce the construction time to $\Oh(n)$.
  This could be difficult, though, since by the results of Section~\ref{sec:by-products}
  this would yield alternative linear-time algorithms finding primitively rooted squares
  and computing seeds.
  The only known linear-time algorithms for these problems
  (see \cite{DBLP:journals/jcss/GusfieldS04,Crochemore2013,DBLP:conf/spire/CrochemoreIKRRW10} and \cite{DBLP:conf/soda/KociumakaKRRW12})
  are rather complex.

  \bibliographystyle{abbrv}
  \bibliography{partial_covers}
\end{document}

%% file: fig1.tex
\begin{tikzpicture}
\def\unitX{3mm}
\def\unitY{3mm}

\foreach \x/\c in {1/b,2/c,3/c,4/c,5/a,6/c,7/c,8/c,9/a,10/c,11/c,12/a,13/c,14/c,15/b} {
    \draw node [above] at (\x*\unitX, 0) {\tt \c}; 
    \draw node [below] at (\x*\unitX, 0) {\tiny \tt \x}; 
}

\foreach \x/\y/\label in {4/1/,8/1/,11/2/} {
    \draw (\x*\unitX-0.3*\unitX,\y*\unitY*0.8+0.4*\unitY) -- (\x*\unitX-0.3*\unitX,\y*\unitY*0.8+0.7*\unitY)--+(3.6*\unitX,0)--+(3.6*\unitX,-0.3*\unitY);
}
\end{tikzpicture}

%% file: fig2_suf_tree.tex
  \begin{tikzpicture}[scale=.85]
      \filldraw (0,-.5) circle (0.07cm);

      \draw (0,-.5) -- node[above=0.15cm,sloped] {\begin{rotate}{270}$c$\end{rotate}\ \ \ \begin{rotate}{270}$c$\end{rotate}\ \ \ \begin{rotate}{270}$a$\end{rotate}} (-4.4,-5);
    \begin{scope}[yshift=-2cm,xshift=-1cm]
      \filldraw (-3.4,-3) node[above left=-0.05cm] {9, 3} circle (0.07cm);
      \filldraw (-5.4,-5.5) node[below=0.05cm] {7, 1} circle (0.07cm);
      \filldraw (-3.8,-5) node[below=0.05cm] {4, 1} circle (0.07cm);
      \filldraw (-2.6,-6.5) node[below=0.05cm] {11, 1} circle (0.07cm);
      \draw (-3.4,-3) -- node[above=0.15cm,sloped] {\begin{rotate}{270}$b$\end{rotate}\ \ \ \begin{rotate}{270}$c$\end{rotate}\ \ \ \begin{rotate}{270}$c$\end{rotate}\ \ \ \begin{rotate}{270}$a$\end{rotate}} (-5.4,-5.5);
      \draw (-3.4,-3) -- node[above=0.15cm,sloped,near end] {\begin{rotate}{270}${b}$\end{rotate}} (-3.8,-5);
      \draw (-3.4,-3) -- node[above,sloped] {\begin{rotate}{90}$c$\end{rotate}\ \ \ \begin{rotate}{90}$a$\end{rotate}\ \ \ \begin{rotate}{90}$c$\end{rotate}\ \ \ \begin{rotate}{90}$c$\end{rotate}\ \ \ \begin{rotate}{90}$a$\end{rotate}\ \ \ \begin{rotate}{90}$c$\end{rotate}\ \ \ \begin{rotate}{90}$c$\end{rotate}\ \ \ \begin{rotate}{90}$b$\end{rotate}} (-2.6,-6.5);
    \end{scope}

    \draw (0,-.5) -- node[above=0.15cm,sloped,near end] {\begin{rotate}{270}$b$\end{rotate}} (-0.2,-1.5);
    \draw (-0.2,-1.5) -- node[above=0.15cm,sloped] {\begin{rotate}{270}$b$\end{rotate}\ \ \ \begin{rotate}{270}$c$\end{rotate}\ \ \ \begin{rotate}{270}$c$\end{rotate}\ \ \ \begin{rotate}{270}$a$\end{rotate}\ \ \ \begin{rotate}{270}$c$\end{rotate}\ \ \ \begin{rotate}{270}$c$\end{rotate}\ \ \ \begin{rotate}{270}$a$\end{rotate}\ \ \ \begin{rotate}{270}$c$\end{rotate}\ \ \ \begin{rotate}{270}$c$\end{rotate}\ \ \ \begin{rotate}{270}$c$\end{rotate}\ \ \ \begin{rotate}{270}$a$\end{rotate}\ \ \ \begin{rotate}{270}$c$\end{rotate}\ \ \ \begin{rotate}{270}$c$\end{rotate}\ \ \ \begin{rotate}{270}$c$\end{rotate}} (-2.5,-7.5);
    \filldraw (-0.2,-1.5) node[right] {2, 2} circle (0.07cm);
    \filldraw (-2.5,-7.5) node[below=0.05cm] {15, 1} circle (0.07cm);

    \draw (0,-.5) -- node[above,sloped] {\begin{rotate}{90}$c$\end{rotate}} (1.5,-1.5);
    \begin{scope}[xshift=1.5cm,yshift=-1.5cm]
      \filldraw (0,0) node[above right=-0.05cm] {10, 10} circle (0.07cm);
      \filldraw (-1,-3) node[left] {11, 2} circle (0.07cm);
      \filldraw (-3,-5.5) node[below=0.05cm] {8, 1} circle (0.07cm);
      \filldraw (-1.4,-5) node[below=0.05cm] {5, 1} circle (0.07cm);
      \filldraw (-0.2,-6.5) node[below left=-0.05cm] {12, 1} circle (0.07cm);
      \draw (0,0) -- node[above=0.15cm,sloped] {\begin{rotate}{270}$c$\end{rotate}\ \ \ \begin{rotate}{270}$a$\end{rotate}} (-0.6,-1.8);
      \filldraw (-0.67,-1.87) rectangle (-0.53,-1.73);
      \draw (-0.6,-1.8) node[left=0.05cm] {9, 3};
      \draw (-0.6,-1.8) -- node[above=0.15cm,sloped] {\begin{rotate}{270}$c$\end{rotate}} (-1,-3);
      \draw (-1,-3) -- node[above=0.15cm,sloped] {\begin{rotate}{270}$b$\end{rotate}\ \ \ \begin{rotate}{270}$c$\end{rotate}\ \ \ \begin{rotate}{270}$c$\end{rotate}\ \ \ \begin{rotate}{270}$a$\end{rotate}} (-3,-5.5);
      \draw (-1,-3) -- node[above=0.15cm,sloped,near end] {\begin{rotate}{270}$b$\end{rotate}} (-1.4,-5);
      \draw (-1,-3) -- node[above,sloped] {\begin{rotate}{90}$c$\end{rotate}\ \ \ \begin{rotate}{90}$a$\end{rotate}\ \ \ \begin{rotate}{90}$c$\end{rotate}\ \ \ \begin{rotate}{90}$c$\end{rotate}\ \ \ \begin{rotate}{90}$a$\end{rotate}\ \ \ \begin{rotate}{90}$c$\end{rotate}\ \ \ \begin{rotate}{90}$c$\end{rotate}\ \ \ \begin{rotate}{90}$b$\end{rotate}} (-0.2,-6.5);
    \end{scope}
    \draw (1.5,-1.5) -- node[above,sloped,near end] {\begin{rotate}{90}$b$\end{rotate}} (2,-3.5);
    \draw (1.5,-1.5) -- node[above,sloped] {\begin{rotate}{90}$c$\end{rotate}} (3.5,-3);
    \filldraw (2,-3.5) node[below=0.05cm] {2, 1} circle (0.07cm);

      \begin{scope}[xshift=3.5cm,yshift=-3cm]
        \filldraw (0,0) node[above right=-0.05cm] {10, 4} circle (0.07cm);
        \filldraw (-0.1,-0.9) rectangle (0.04,-0.76);
        \draw (-0.03,-0.83) node[left=0.05cm] {9, 3};
        \filldraw (-0.13,-1.73) rectangle (0.01,-1.59);
        \draw (-0.06,-1.66) node[left=0.05cm] {11, 2};
        \filldraw (-0.09,-2.5) node[left] {12, 1} circle (0.07cm);
        \filldraw (-1.7,-5) node[below right=-0.05cm] {9, 1} circle (0.07cm);
        \filldraw (-0.2,-4.5) node[below=0.05cm] {6, 1} circle (0.07cm);
        \filldraw (1.1,-5.5) node[below=0.05cm] {13, 1} circle (0.07cm);
        \draw (0,0) -- node[above=0.15cm,sloped] {\begin{rotate}{270}$a$\end{rotate}} (-0.03,-0.83);
        \draw (-0.03,-0.83) -- node[above=0.15cm,sloped] {\begin{rotate}{270}$c$\end{rotate}} (-0.06,-1.66);
        \draw (-0.06,-1.66) -- node[above=0.15cm,sloped] {\begin{rotate}{270}$c$\end{rotate}} (-0.09,-2.5);
        \draw (-0.09,-2.5) -- node[above=0.15cm,sloped] {\begin{rotate}{270}$b$\end{rotate}\ \ \ \begin{rotate}{270}$c$\end{rotate}\ \ \ \begin{rotate}{270}$c$\end{rotate}\ \ \ \begin{rotate}{270}$a$\end{rotate}} (-1.7,-5);
        \draw (-0.09,-2.5) -- node[above=0.15cm,sloped,near end] {\begin{rotate}{270}$b$\end{rotate}} (-0.2,-4.5);
        \draw (-0.09,-2.5) -- node[above,sloped] {\begin{rotate}{90}$c$\end{rotate}\ \ \ \begin{rotate}{90}$a$\end{rotate}\ \ \ \begin{rotate}{90}$c$\end{rotate}\ \ \ \begin{rotate}{90}$c$\end{rotate}\ \ \ \begin{rotate}{90}$a$\end{rotate}\ \ \ \begin{rotate}{90}$c$\end{rotate}\ \ \ \begin{rotate}{90}$c$\end{rotate}\ \ \ \begin{rotate}{90}$b$\end{rotate}} (1.1,-5.5);
      \end{scope}
      
      \draw (3.5,-3) -- node[above,sloped,near end] {\begin{rotate}{90}$b$\end{rotate}} (4.3,-5);
      \filldraw (4.3,-5) node[below=0.05cm] {3, 1} circle (0.07cm);
      \draw (3.5,-3) -- node[above,sloped] {\begin{rotate}{90}$c$\end{rotate} \ \ \begin{rotate}{90}$a$\end{rotate}} (4.75,-4);
      \filldraw (4.68,-4.07) rectangle (4.82,-3.93);
      \draw (4.75,-4) node[above right=-0.05cm] {8, 2};
      \draw (4.75,-4) -- node[above,sloped] {\begin{rotate}{90}$c$\end{rotate} \ \ \begin{rotate}{90}$c$\end{rotate}} (6,-5);
      \filldraw (6,-5) node[above right=-0.05cm] {10, 1} circle (0.07cm);
      \draw (6,-5) -- node[above=0.15cm,sloped] {\begin{rotate}{270}$b$\end{rotate}\ \ \ \begin{rotate}{270}$c$\end{rotate}\ \ \ \begin{rotate}{270}$c$\end{rotate}\ \ \ \begin{rotate}{270}$a$\end{rotate}} (5.4,-8);
      \draw (6,-5) -- node[above,sloped] {\begin{rotate}{90}$c$\end{rotate}\ \ \ \begin{rotate}{90}$a$\end{rotate}\ \ \ \begin{rotate}{90}$c$\end{rotate}\ \ \ \begin{rotate}{90}$c$\end{rotate}\ \ \ \begin{rotate}{90}$a$\end{rotate}\ \ \ \begin{rotate}{90}$c$\end{rotate}\ \ \ \begin{rotate}{90}$c$\end{rotate}\ \ \ \begin{rotate}{90}$b$\end{rotate}} (6.6,-8);
      \filldraw (5.4,-8) node[below=0.05cm] {10, 1} circle (0.07cm);
      \filldraw (6.6,-8) node[below=0.05cm] {14, 1} circle (0.07cm);

  \end{tikzpicture}

%% file: fig5_envelope.tex
  \begin{tikzpicture}[scale=.35]
    \draw[->,-latex] (0,0) -- (16,0);
    \draw[->,-latex] (0,0) -- (0,16);
    \foreach \i in {1,...,15}
      \draw (\i,-0.2) -- node[below] {\scriptsize \i} (\i,0.2);
    \foreach \i in {1,...,15}
      \draw (-0.2,\i) -- node[left] {\scriptsize \i} (0.2,\i);
    \foreach \pos in {(1,10), (2,10), (3,9), (4,11), (5,12), (6,10), (7,7), (8,8), (9,9), (10,10), (11,11), (12,12), (13,13), (14,14), (15,15)}
      \draw \pos circle (0.15cm);
    \draw[thick,densely dashed]
      (1,10) -- (2,10) -- (2,9) -- (3,9) -- (3,11) -- (4,11) -- (4,12) -- (5,12)
      -- (5,10) -- (6,10) -- (6,7) -- (7,7) -- (7,8) -- (8,8) -- (8,9) -- (9,9)
      -- (9,10) -- (10,10) -- (10,11) -- (11,11) -- (11,12) -- (12,12) -- (12,13)
      -- (13,13) -- (13,14) -- (14,14) -- (14,15) -- (15,15);

    \draw (16,0.5) node[right] {word};
    \draw (16,-0.5) node[right] {length};
    \draw (-4.5,16.8) node[right] {covered};
    \draw (-4.5,15.8) node[right] {positions};

\draw (3,9) -- (1,3);
\filldraw (3,9) circle (0.07cm);
\filldraw (1,3) circle (0.07cm);
\draw (7,7) -- (4,4);
\filldraw (7,7) circle (0.07cm);
\filldraw (4,4) circle (0.07cm);
\draw (4,4) -- (4,4);
\filldraw (4,4) circle (0.07cm);
\draw (11,11) -- (4,4);
\filldraw (11,11) circle (0.07cm);
\filldraw (4,4) circle (0.07cm);
\draw (1,2) -- (1,2);
\filldraw (1,2) circle (0.07cm);
\draw (15,15) -- (2,2);
\filldraw (15,15) circle (0.07cm);
\filldraw (2,2) circle (0.07cm);
\draw (1,10) -- (1,10);
\filldraw (1,10) circle (0.07cm);
\draw (3,9) -- (2,6);
\filldraw (3,9) circle (0.07cm);
\filldraw (2,6) circle (0.07cm);
\draw (4,11) -- (4,11);
\filldraw (4,11) circle (0.07cm);
\draw (8,8) -- (5,5);
\filldraw (8,8) circle (0.07cm);
\filldraw (5,5) circle (0.07cm);
\draw (5,5) -- (5,5);
\filldraw (5,5) circle (0.07cm);
\draw (12,12) -- (5,5);
\filldraw (12,12) circle (0.07cm);
\filldraw (5,5) circle (0.07cm);
\draw (2,2) -- (2,2);
\filldraw (2,2) circle (0.07cm);
\draw (2,10) -- (2,10);
\filldraw (2,10) circle (0.07cm);
\draw (3,9) -- (3,9);
\filldraw (3,9) circle (0.07cm);
\draw (4,11) -- (4,11);
\filldraw (4,11) circle (0.07cm);
\draw (5,12) -- (5,12);
\filldraw (5,12) circle (0.07cm);
\draw (9,9) -- (6,6);
\filldraw (9,9) circle (0.07cm);
\filldraw (6,6) circle (0.07cm);
\draw (6,6) -- (6,6);
\filldraw (6,6) circle (0.07cm);
\draw (13,13) -- (6,6);
\filldraw (13,13) circle (0.07cm);
\filldraw (6,6) circle (0.07cm);
\draw (3,3) -- (3,3);
\filldraw (3,3) circle (0.07cm);
\draw (4,8) -- (3,6);
\filldraw (4,8) circle (0.07cm);
\filldraw (3,6) circle (0.07cm);
\draw (6,10) -- (5,9);
\filldraw (6,10) circle (0.07cm);
\filldraw (5,9) circle (0.07cm);
\draw (10,10) -- (7,7);
\filldraw (10,10) circle (0.07cm);
\filldraw (7,7) circle (0.07cm);
\draw (14,14) -- (7,7);
\filldraw (14,14) circle (0.07cm);
\filldraw (7,7) circle (0.07cm);
  \end{tikzpicture}

%% file: fig4_merge_list.tex
    \begin{tikzpicture}[scale=.7]
      \def\unitX{1cm}
      \def\unitY{1cm}

      \begin{scope}[xshift=3cm,yshift=2cm]
        \draw[thick,densely dotted,->,-latex] (1.07,0.07) sin (1.5,0.3) cos (1.93,0.07);
        \draw[xshift=1cm,thick,densely dotted,->,-latex] (1.07,0.07) sin (1.5,0.3) cos (1.93,0.07);
        \draw[xshift=3cm,thick,densely dotted,->,-latex] (1.07,0.07) sin (1.5,0.3) cos (1.93,0.07);
        \draw[xshift=6cm,thick,densely dotted,->,-latex] (1.07,0.07) sin (1.5,0.3) cos (1.93,0.07);

        \filldraw (1,0) node[below=0.05cm] {1} circle (0.1cm);
        \draw (2,0) node[below=0.05cm] {2} circle (0.1cm);
        \filldraw (3,0) node[below=0.05cm] {3} circle (0.1cm);
        \filldraw (4,0) node[below=0.05cm] {4} circle (0.1cm);
        \draw (5,0) node[below=0.05cm] {5} circle (0.1cm);
        \draw (6,0) node[below=0.05cm] {6} circle (0.1cm);
        \draw (7,0) node[below=0.05cm] {7} circle (0.1cm);
        \filldraw[white!70!black] (8,0) circle (0.1cm);
        \draw (8,0) node[below=0.05cm] {8} circle (0.1cm);
        \filldraw[white!70!black] (9,0) circle (0.1cm);
        \draw (9,0) node[below=0.05cm] {9} circle (0.1cm);
      \end{scope}

      \begin{scope}[yshift=4.5cm]
        \draw (0,0) node[above] {$v$};
        \draw (-2.5,-2) node[above left] {$v_1$};
        \draw (0,-1.5) node[above left] {$v_2$};
        \draw (2.25,-1.3) node[above] {$v_3$};

        \draw (0,0) -- (0,-1.5) -- (-1,-2.25) -- (-1.5,-3.25);
        \draw (0,-1.5) -- (1,-2.25) -- (1.5,-3.25);
        \draw (-1,-2.25) -- (-0.5,-3.25);
        \draw (1,-2.25) -- (0.5,-3.25);
        \filldraw[white] (-1.5,-3.25) circle (0.1cm);
        \draw (-1.5,-3.25) node[below=0.05cm] {7} circle (0.1cm);
        \filldraw[white] (-0.5,-3.25) circle (0.1cm);
        \draw (-0.5,-3.25) node[below=0.05cm] {5} circle (0.1cm);
        \filldraw[white] (0.5,-3.25) circle (0.1cm);
        \draw (0.5,-3.25) node[below=0.05cm] {2} circle (0.1cm);
        \filldraw[white] (1.5,-3.25) circle (0.1cm);
        \draw (1.5,-3.25) node[below=0.05cm] {6} circle (0.1cm);

        \draw (0,0) -- (-2.5,-2) -- (-3,-3);
        \draw (-2.5,-2) -- (-2,-3);
        \filldraw[white!70!black] (-3,-3) circle (0.1cm);
        \draw (-3,-3) node[below=0.05cm] {9} circle (0.1cm);
        \filldraw[white!70!black] (-2,-3) circle (0.1cm);
        \draw (-2,-3) node[below=0.05cm] {8} circle (0.1cm);

        \draw (0,0) -- (2.25,-1.3) -- (2.75,-2.25) -- (2.25,-3.25);
        \draw (2.25,-1.3) -- (1.75,-2.25);
        \draw (2.75,-2.25) -- (3.25,-3.25);
        \filldraw (1.75,-2.25) node[below=0.05cm] {4} circle (0.1cm);
        \filldraw (2.25,-3.25) node[below=0.05cm] {3} circle (0.1cm);
        \filldraw (3.25,-3.25) node[below=0.05cm] {1} circle (0.1cm);
      \end{scope}

    \end{tikzpicture}

%% file: drzewko.tex
    \begin{tikzpicture}[scale=.8]
      \def\unitX{1cm}
      \def\unitY{1cm}
      
      \newcommand{\subtree}[4]{
        \draw (#1*\unitX,#2*\unitY)--+(0.5*\unitX,-1*\unitY)--+(-0.5*\unitX,-1*\unitY)--cycle;
        \filldraw[black,fill=white] (#1*\unitX,#2*\unitY) circle (\unitX/20);
        \draw node (#4) at (#1*\unitX,#2*\unitY) {};
        \draw node [#3] at (#1*\unitX,#2*\unitY) {$#4$};
      }
      
      \subtree{0}{0}{above left}{v_1}
      \subtree{2}{0}{above left}{v_2}
      \subtree{4}{0.2}{above left}{v_3}
      \subtree{6}{0.1}{above right}{v_4}
      \subtree{8}{0.5}{above right}{v_5}

      \draw let \p1=(v_3),\p2=(v_4) in
        ($(\x1/2+\x2/2,\y1+\unitY)$) node (v) {};
      \draw node at (v) [above] {$\mathit{v}$};
      \draw [dashed, shorten >= -0.1cm] (v)--(v_3);
      \draw [dashed, shorten >= -0.1cm] (v)--(v_4);
      \draw (v) circle (\unitX/20);

      \draw let \p1=(v_1),\p2=(v_5) in
        ($(\x1/2+\x2/2,\y1+3*\unitY)$) node (root) {};
      \draw node at (root) [above] {$\mathit{root}$};
      \draw [dashed, shorten >= -0.1cm] (root)--(v_1);
      \draw [dashed, shorten >= -0.07cm] (root)--(v_5);

      \draw [dashed,<->] let \p1=(root),\p2=(v),\p3=(1*\unitX,0) in
        (\x3,\y1)--(\x3,\y2) node[midway, left] {$h$};;

      \filldraw [black] let \p1=(v_3) in 
        (\x1-\unitX*0.3,\y1-\unitY) circle (\unitX/20) node [below] {$i$};

    \end{tikzpicture}

%% file: fig6_seed.tex
\begin{tikzpicture}[scale=0.9]
{\small
 \draw (0,0) node[above] {\texttt{a}};
  \draw (0.5,0) node[above] {\texttt{b}};
  \draw (1,0) node[above] {\texttt{a}};
  \draw (1.5,0) node[above] {\texttt{a}};
  \draw (2,0) node[above] {\texttt{b}};
  \draw (2.5,0) node[above] {\texttt{a}};
  \draw (3,0) node[above] {\texttt{a}};
  \draw (3.5,0) node[above] {\texttt{b}};
  \draw (4,0) node[above] {\texttt{a}};
  \draw (4.5,0) node[above] {\texttt{a}};
  \draw (5,0) node[above] {\texttt{a}};
  \draw (5.5,0) node[above] {\texttt{b}};
  \draw (6,0) node[above] {\texttt{a}};
  \draw (6.5,0) node[above] {\texttt{a}};
  \draw (-1,0) node[above] {\texttt{a}};
  \draw (-0.5,0) node[above] {\texttt{a}};
  \draw (7,0) node[above] {\texttt{b}};
  \draw (7.5,0) node[above] {\texttt{a}};
  \draw[xshift=-2cm] (-0.25,0.4) .. controls (0.5,1.2) and (1,1.2) .. (1.75,0.4);
  \filldraw[white] (-2.25,0) rectangle (-1.25,1.5);
  \draw[xshift=1.5cm] (-0.25,0.4) .. controls (0.5,1.2) and (1,1.2) .. (1.75,0.4);
  \draw[yshift=0.1cm] (-0.25,0) .. controls (0.5,-0.8) and (1,-0.8) .. (1.75,0);
  \draw[xshift=5cm] (-0.25,0.4) .. controls (0.5,1.2) and (1,1.2) .. (1.75,0.4);
  \draw[xshift=3cm,yshift=0.1cm] (-0.25,0) .. controls (0.5,-0.8) and (1,-0.8) .. (1.75,0);
  \draw[xshift=6.5cm,yshift=0.1cm] (-0.25,0) .. controls (0.5,-0.8) and (1,-0.8) .. (1.75,0);
  \filldraw[white,yshift=0.1cm] (7.75,0) rectangle (8.25,-1.5);
}
\end{tikzpicture}

%% file: partial_covers.bbl
\begin{thebibliography}{10}

\bibitem{Apo93}
A.~Apostolico and A.~Ehrenfeucht.
\newblock Efficient detection of quasiperiodicities in strings.
\newblock {\em Theor. Comput. Sci.}, 119(2):247--265, 1993.

\bibitem{DBLP:journals/ipl/ApostolicoFI91}
A.~Apostolico, M.~Farach, and C.~S. Iliopoulos.
\newblock Optimal superprimitivity testing for strings.
\newblock {\em Inf. Process. Lett.}, 39(1):17--20, 1991.

\bibitem{DBLP:journals/algorithmica/ApostolicoP96}
A.~Apostolico and F.~P. Preparata.
\newblock Data structures and algorithms for the string statistics problem.
\newblock {\em Algorithmica}, 15(5):481--494, 1996.

\bibitem{DBLP:journals/ipl/Breslauer92}
D.~Breslauer.
\newblock An on-line string superprimitivity test.
\newblock {\em Inf. Process. Lett.}, 44(6):345--347, 1992.

\bibitem{DBLP:conf/icalp/BrodalLOP02}
G.~S. Brodal, R.~B. Lyngs{\o}, A.~{\"O}stlin, and C.~N.~S. Pedersen.
\newblock Solving the string statistics problem in time ${O}(n \log n)$.
\newblock In P.~Widmayer, F.~T. Ruiz, R.~M. Bueno, M.~Hennessy, S.~Eidenbenz,
  and R.~Conejo, editors, {\em ICALP}, volume 2380 of {\em Lecture Notes in
  Computer Science}, pages 728--739. Springer, 2002.

\bibitem{DBLP:conf/cpm/BrodalP00}
G.~S. Brodal and C.~N.~S. Pedersen.
\newblock Finding maximal quasiperiodicities in strings.
\newblock In R.~Giancarlo and D.~Sankoff, editors, {\em CPM}, volume 1848 of
  {\em Lecture Notes in Computer Science}, pages 397--411. Springer, 2000.

\bibitem{DBLP:journals/jacm/BrownT79}
M.~R. Brown and R.~E. Tarjan.
\newblock A fast merging algorithm.
\newblock {\em J. ACM}, 26(2):211--226, 1979.

\bibitem{AlgorithmsOnStrings}
M.~Crochemore, C.~Hancart, and T.~Lecroq.
\newblock {\em Algorithms on Strings}.
\newblock Cambridge University Press, 2007.

\bibitem{DBLP:journals/tcs/CrochemoreIR09}
M.~Crochemore, L.~Ilie, and W.~Rytter.
\newblock Repetitions in strings: Algorithms and combinatorics.
\newblock {\em Theor. Comput. Sci.}, 410(50):5227--5235, 2009.

\bibitem{Crochemore2013}
M.~Crochemore, C.~Iliopoulos, M.~Kubica, J.~Radoszewski, W.~Rytter, and
  T.~Wale\'n.
\newblock Extracting powers and periods in a word from its runs structure.
\newblock {\em Theoretical Computer Science}, doi: 10.1016/j.tcs.2013.11.018,
  2013.

\bibitem{DBLP:conf/spire/CrochemoreIKRRW10}
M.~Crochemore, C.~S. Iliopoulos, M.~Kubica, J.~Radoszewski, W.~Rytter, and
  T.~Wale\'n.
\newblock Extracting powers and periods in a string from its runs structure.
\newblock In E.~Ch{\'a}vez and S.~Lonardi, editors, {\em SPIRE}, volume 6393 of
  {\em Lecture Notes in Computer Science}, pages 258--269. Springer, 2010.

\bibitem{Jewels}
M.~Crochemore and W.~Rytter.
\newblock {\em Jewels of Stringology}.
\newblock World Scientific, 2003.

\bibitem{DBLP:conf/focs/Farach97}
M.~Farach.
\newblock Optimal suffix tree construction with large alphabets.
\newblock In {\em FOCS}, pages 137--143, 1997.

\bibitem{Flouri2013102}
T.~Flouri, C.~S. Iliopoulos, T.~Kociumaka, S.~P. Pissis, S.~J. Puglisi,
  W.~Smyth, and W.~Tyczy\'nski.
\newblock Enhanced string covering.
\newblock {\em Theoretical Computer Science}, 506(0):102 -- 114, 2013.

\bibitem{DBLP:journals/jct/FraenkelS98}
A.~S. Fraenkel and J.~Simpson.
\newblock How many squares can a string contain?
\newblock {\em J. Comb. Theory, Ser. A}, 82(1):112--120, 1998.

\bibitem{DBLP:journals/jcss/GusfieldS04}
D.~Gusfield and J.~Stoye.
\newblock Linear time algorithms for finding and representing all the tandem
  repeats in a string.
\newblock {\em J. Comput. Syst. Sci.}, 69(4):525--546, 2004.

\bibitem{DBLP:journals/ipl/Hershberger89}
J.~Hershberger.
\newblock Finding the upper envelope of n line segments in {O}(n log n) time.
\newblock {\em Inf. Process. Lett.}, 33(4):169--174, 1989.

\bibitem{DBLP:journals/algorithmica/IliopoulosMP96}
C.~S. Iliopoulos, D.~W.~G. Moore, and K.~Park.
\newblock Covering a string.
\newblock {\em Algorithmica}, 16(3):288--297, 1996.

\bibitem{DBLP:conf/soda/KociumakaKRRW12}
T.~Kociumaka, M.~Kubica, J.~Radoszewski, W.~Rytter, and T.~Wale\'n.
\newblock A linear time algorithm for seeds computation.
\newblock In Y.~Rabani, editor, {\em SODA}, pages 1095--1112. SIAM, 2012.

\bibitem{DBLP:conf/focs/KolpakovK99}
R.~M. Kolpakov and G.~Kucherov.
\newblock Finding maximal repetitions in a word in linear time.
\newblock In {\em FOCS}, pages 596--604. IEEE Computer Society, 1999.

\bibitem{DBLP:conf/cpm/KucherovNS12}
G.~Kucherov, Y.~Nekrich, and T.~A. Starikovskaya.
\newblock Cross-document pattern matching.
\newblock In J.~K{\"a}rkk{\"a}inen and J.~Stoye, editors, {\em CPM}, volume
  7354 of {\em Lecture Notes in Computer Science}, pages 196--207. Springer,
  2012.

\bibitem{DBLP:journals/algorithmica/LiS02}
Y.~Li and W.~F. Smyth.
\newblock Computing the cover array in linear time.
\newblock {\em Algorithmica}, 32(1):95--106, 2002.

\bibitem{DBLP:journals/ipl/MooreS94}
D.~Moore and W.~F. Smyth.
\newblock An optimal algorithm to compute all the covers of a string.
\newblock {\em Inf. Process. Lett.}, 50(5):239--246, 1994.

\bibitem{SimParkKimLee}
J.~S. Sim, K.~Park, S.~Kim, and J.~Lee.
\newblock Finding approximate covers of strings.
\newblock {\em Journal of Korea Information Science Society}, 29(1):16--21,
  2002.

\bibitem{DBLP:journals/algorithmica/Ukkonen95}
E.~Ukkonen.
\newblock On-line construction of suffix trees.
\newblock {\em Algorithmica}, 14(3):249--260, 1995.

\end{thebibliography}
